\documentclass[a4paper,UKenglish,cleveref, autoref, thm-restate]{lipics-v2021}

\pdfoutput=1 
\hideLIPIcs  

\usepackage{cite}
\usepackage{amsfonts,amssymb,amsmath,dsfont,amsthm}
\usepackage{xspace} 
\usepackage{graphicx}
\usepackage{multirow}
\usepackage[table]{xcolor}
\usepackage{bm}
\usepackage{array}
\newcolumntype{C}[1]{>{\centering\let\newline\\\arraybackslash\hspace{0pt}}m{#1}}

\newcommand{\tvect}[3]{\ensuremath{\Bigl(\negthinspace\begin{smallmatrix}#1\\#2\\#3\end{smallmatrix}\Bigr)}}

\newcommand{\bigO}{\mathcal{O}}

\newcommand{\integers}{\mathds{Z}}

\newcommand*{\eg}{\textit{e.g.}\@\xspace}
\newcommand*{\ie}{\textit{i.e.}\@\xspace}
\newcommand*{\aka}{\textit{a.k.a.}\@\xspace}
\newcommand*{\etal}{\textit{et al.}\@\xspace}
\newcommand*{\resp}{\textit{resp.}\@\xspace}
\newcommand{\indic}{\mathds{1}}
\newcommand{\vect}[1]{\bm{#1}}

\newcommand{\vc}{\vect{c}}
\newcommand{\vx}{\vect{x}}
\newcommand{\vy}{\vect{y}}
\newcommand{\ve}{\vect{e}}
\newcommand{\vz}{\vect{z}}

\newcommand{\cc}{\operatorname{colour}}
\newcommand{\vv}{V}

\newcommand{\window}{\operatorname{Window}}
\newcommand{\multiset}{\operatorname{colourMultiset}}
\newcommand{\dual}{\operatorname{Dual}}
\newcommand{\profile}{\operatorname{Profile}}

\newcommand{\mA}{\mathcal{A}}
\newcommand{\ka}{\textcolor{red}{0}}
\newcommand{\kb}{\textcolor{green}{1}}
\newcommand{\kc}{\textcolor{blue}{2}}

\title{Robot~Positioning~Using~Torus~Packing~for~Multisets}

\author{Chung Shue {Chen}\footnote{Authors presented in alphabetical order.}}{Nokia Bell Labs, France}{chung_shue.chen@nokia-bell-labs.com}{}{}

\author{Peter {Keevash}}{Mathematical Institute, University of Oxford, UK}{peter.keevash@maths.ox.ac.uk}{}{Supported by ERC Advanced Grant 883810}

\author{Sean {Kennedy}}{Nokia Bell Labs, Canada}{sean.kennedy@nokia-bell-labs.com}{}{}

\author{\'Elie {de Panafieu}}{Nokia Bell Labs, France}{elie.de_panafieu@nokia-bell-labs.com}{}{}

\author{Adrian {Vetta}}{McGill University, Canada}{adrian.vetta@mcgill.ca}{}{}

\authorrunning{C.\,S Chen, P. Keevash, S. Kennedy, E. de Panafieu and A. Vetta} 

\Copyright{Chung Shue Chen, Peter Keevash, Sean Kennedy, \'Elie de Panafieu, and Adrian Vetta} 



\category{} 

\relatedversion{} 

\nolinenumbers 

\begin{document}

\maketitle
%
\begin{abstract}
We consider the design of a positioning system where a robot determines its position from local observations. This is a well-studied problem of considerable practical importance and mathematical interest. The dominant paradigm derives from the classical theory of de Bruijn sequences, where the robot has access to a window within a larger code and can determine its position if these windows are distinct. We propose an alternative model in which the robot has more limited observational powers, which we argue is more realistic in terms of engineering: the robot does not have access to the full pattern of colours (or letters) in the window, but only to the intensity of each colour (or the number of occurrences of each letter). This leads to a mathematically interesting problem with a different flavour to that arising in the classical paradigm, requiring new construction techniques. The parameters of our construction are optimal up to a constant factor, and computing the position requires only a constant number of arithmetic operations.
\end{abstract}


\section{Introduction}

Consider a robot located on a grid of coloured squares that must determine
its position after observing part of the grid through a fixed viewing window.
The dominant paradigm for this problem derives from the mathematical theory
of de Bruijn sequences, i.e.\ binary (or bi-coloured) cyclic sequences of length $2^n$ 
in which each binary sequence of length $n$ 
appears exactly once as a subsequence of consecutive entries.
Such a sequence can be used for positioning a robot in one dimension:
the robot sees a viewing window of length $n$; this window induces a subsequence with a unique colour pattern;
from this colour pattern the robot can then reconstruct its position in the sequence (if we disregard issues of computational efficiency and error correction).
Generalisations of this idea to higher dimensions and related combinatorial
structures have led to a rich mathematical theory; see \cref{sec:rel}.

However, while this theory is mathematically pleasing, we will argue in \cref{sec:eng}
and Appendix~A
that engineering constraints support a model in which the robot
does not have access to the full colour pattern in the window,
but must infer its position only knowing the intensity of each colour,
that is, the multiset of colours.
More precisely, given an $n \times n$ grid,
a robot with an $m \times m$ viewing window, and a palette of $k$ colours, our task is to colour the grid so that each possible location of the viewing windows produces a different multiset of colours. Furthermore, it will be mathematically more natural to undertake this task for a torus of side $n$ rather than a grid, and also to generalise to an arbitrary dimension $d$.

\begin{example}
The grid colouring illustrated in Figure~\ref{fig:grid:coloring} has dimension $d=2$, size $n=8$, window size $m=4$, and $k=3$ colours (red $\ka$, green $\kb$, and blue $\kc$). No two $4 \times 4$ subsquares contain the same multiset of colours. For example, two viewing windows are shown with multisets that have multiplicities for $(\ka,\kb,\kc)$ equal to $(6,2,8)$ and $(3,5,8)$, respectively. Observe that the second window ``\emph{wraps around}'' because the grid is considered as a torus.
\end{example}

\begin{figure}[h!]
\begin{center}
$\begin{array}{|c|c|c|c|c|c|c|c|}
\hline
\kc & \kc & \kc & \kc & \kc & \kb & \kc & \kb
\\ \hline
\kc & \kc & \kc & \kc & \kb & \kc & \kb & \kc
\\ \hline
\kc & \kc & \kc & \kc & \kc & \kb & \kc & \kc
\\ \hline
\kc & \kc & \kc & \kc & \kb & \kc & \kb & \kc
\\ \hline
\ka & \kc & \ka & \kc & \ka & \kb & \ka & \kb
\\ \hline
\kc & \ka & \kc & \ka & \kb & \ka & \kb & \ka
\\ \hline
\ka & \kc & \ka & \kc & \ka & \kb & \ka & \kc
\\ \hline
\kc & \ka & \kc & \kc & \kb & \ka & \kb & \kc
\\ \hline
\end{array}
\quad
\begin{array}{|c|c|c|c|c|c|c|c|}
\hline
\kc & \kc & \kc & \kc & \kc & \kb & \kc & \kb
\\ \hline
\kc & \kc & \kc & \kc & \kb & \kc & \kb & \kc
\\ \hline
\kc & \kc & \kc & \kc & \kc & \kb & \kc & \kc
\\ \hline
\kc & \cellcolor{lightgray}\kc & \cellcolor{lightgray}\kc & \cellcolor{lightgray}\kc & \cellcolor{lightgray}\kb & \kc & \kb & \kc
\\ \hline
\ka & \cellcolor{lightgray}\kc & \cellcolor{lightgray}\ka & \cellcolor{lightgray}\kc & \cellcolor{lightgray}\ka & \kb & \ka & \kb
\\ \hline
\kc & \cellcolor{lightgray}\ka & \cellcolor{lightgray}\kc & \cellcolor{lightgray}\ka & \cellcolor{lightgray}\kb & \ka & \kb & \ka
\\ \hline
\ka & \cellcolor{lightgray}{\kc} & \cellcolor{lightgray}\ka & \cellcolor{lightgray}\kc & \cellcolor{lightgray}\ka & \kb & \ka & \kc
\\ \hline
\kc & \ka & \kc & \kc & \kb & \ka & \kb & \kc
\\ \hline
\end{array}
\quad
\begin{array}{|c|c|c|c|c|c|c|c|}
\hline
\cellcolor{lightgray}\kc & \kc & \kc & \kc & \kc & \cellcolor{lightgray}\kb & \cellcolor{lightgray}\kc & \cellcolor{lightgray}\kb
\\ \hline
\cellcolor{lightgray}\kc & \kc & \kc & \kc & \kb & \cellcolor{lightgray}\kc & \cellcolor{lightgray}\kb & \cellcolor{lightgray}\kc
\\ \hline
\kc & \kc & \kc & \kc & \kc & \kb & \kc & \kc
\\ \hline
\kc & \kc & \kc & \kc & \kb & \kc & \kb & \kc
\\ \hline
\ka & \kc & \ka & \kc & \ka & \kb & \ka & \kb
\\ \hline
\kc & \ka & \kc & \ka & \kb & \ka & \kb & \ka
\\ \hline
\cellcolor{lightgray}\ka & \kc & \ka & \kc & \ka & \cellcolor{lightgray}\kb & \cellcolor{lightgray}\ka & \cellcolor{lightgray}\kc
\\ \hline
\cellcolor{lightgray}\kc & \ka & \kc & \kc & \kb & \cellcolor{lightgray}\ka & \cellcolor{lightgray}\kb & \cellcolor{lightgray}\kc
\\ \hline
\end{array}
$
\caption{}
\label{fig:grid:coloring}
\end{center}
\end{figure}

Our question of study is, given $m,d,k$, what is the largest grid size $n$ for which position reconstruction is possible?
Unfortunately, even the one-dimensional version of this question is the subject of several unsolved problems,
discussed in \cref{sec:rel}.
However, for practical purposes one can be content to relax from finding
the optimal $n$ given $m,d,k$ to a value that is optimal up to a constant factor,
where we think of $d$ and $k$ as fixed and consider the asymptotics for large $m$ and~$n$.
There is a clear information theoretic barrier (see \Cref{th:bound:size}) at $n = \Theta_{k,d}(m^{k-1})$, where the subscripts indicate that $k$ and $d$ are constants.
The main contribution of this paper is a construction that achieves this theoretical optimum up to a constant factor, and moreover has optimal computational efficiency, in that only a constant number of arithmetic operations are required to compute the location of the window from its multiset of colours.
We implemented and tested this construction in Python~\cite{ourcode}.

    \subsection{Definitions and Results}

\noindent \textbf{Notations.}
All vectors, tuples and sequences are indexed starting at $0$.
Integer intervals are denoted with square brackets, such as $[0, n-1]$.
Given a length $d$, the vector $\ve_i$
has all coordinates equal to~$0$, except the $i$th,
which is equal to~$1$.
The value of $\indic_{\text{condition}}$ is $1$
if \emph{condition} is satisfied, and $0$ otherwise.
We write $i \equiv j \bmod n$
when the integer $i - j$ is divisible by $n$.
Tables are represented with row indices
increasing from top to bottom and column indices
increasing from left to right, both starting at $0$.

\medskip

We first present general definitions
that will be useful for discussing the literature
and presenting our construction,
starting with \emph{cycle packings}
for the one-dimensional case,
and continuing with their higher dimensional generalisation,
\emph{torus packings}.
Then the central objects of this article,
\emph{grid colouring}, are formally defined
as particular cases of torus packings.
Our main result, \cref{th:construction:grid:colouring},
is a near-optimal construction for grid colourings.
It will make use of \emph{vector sum packings}
(\cref{def:vector sum packing}),
which are particular cases of cycle packings.

\begin{definition} \label{def:cycle:packing}
Given an alphabet $\mA$,
a size $n$,
a window size $m$,
and a function $f$ on $\mA^m$,
an $(\mA, f, n, m)$-\textbf{cycle packing}
is a function $W : \integers \mapsto \mA$
that satisfies
\begin{itemize}
\item {\tt Periodicity.}
For all $x \in \integers$ 
we have $W_x = W_{x + n}$.
\item {\tt Injectivity.}
If
$f(W_x, W_{x+1}, \ldots, W_{x+m-1}) = f(W_y, W_{y+1}, \ldots, W_{y+m-1})$, for any integers $x$ and $y$,
then $x \equiv y \bmod n$.
\end{itemize}
\end{definition}

Let us explain this formal definition.
Consider a circle composed of $n$ squares,
each receiving a letter from $\mA$.
A robot located on this circle wants to recover its position.
It makes a local measurement,
function $f$ of the letters in a window of size $m$ around it.
The injectivity condition ensures
that two positions of the robot (and thus of the window)
correspond to two distinct values of $f$.
Thus the robot is always able to deduce its position.
Observe that the form of function $f$ matters as different functions induce different types of information that robot can extract from the viewing window.
To illustrate this point, consider the following three examples of
$(\mA, f, n, m)$-cycle packings. These examples all have
$\mA = \{\ka,\kb, \kc\}$ and $m=3$ but differ in their functions $f$
(which in turn lead to varying sizes of $n$).

\begin{example}
Take the basic case where $f$ is the identity function.
The robot thus receives the entire colour pattern in its viewing window. Thus a cycle packing corresponds to a sequence of length $n$ where all the contiguous subsequences of length $m$ have distinct colour patterns.
With $\mA = \{\ka,\kb, \kc\}$
and window size $m=3$, we have that
$(\ka,\kb,\kb,\kb,\kc,\kb,\ka,\kb,\kc,\ka,\kb,\ka,\kc,\kb,\kb,
\ka,\ka,\kc,\kc,\kb,\kc,\kc,\kc,\ka,\kc,\ka,\ka)$
is a cycle-packing of size $n=27$ as
each possible string of length $3$ appears at most once; thus, 
the injectivity property holds.
In fact, each such string appears {\em exactly} once (\eg the factor $(\ka, \ka, \ka)=f(\ka, \ka, \ka)$ appears by wrapping around),
making it a de Bruijn sequence. This exactness property is not required for cycle packings, but when it is satisfied
the cycle packing is called a \emph{universal cycle}.
\end{example}

\begin{example}
Recall our motivation is that the robot extracts the colour intensities rather than the entire colour pattern.
So instead of the identity function, assume that $f$ is the multiset counting function which simply counts the number of appearances of each colour in the viewing window. 
Now, for the alphabet $\mA = \{\ka,\kb, \kc\}$
with window size $m=3$, the sequence $(\ka,\kb,\kb,\kb,\kc,\kc,\kc,\ka,\ka)$ is a cycle packing of size $n=9$ as every multiset appears at most once (\eg the multiset $(3, 0, 0)=f(\ka,\ka,\ka)$ arises by wrapping around) and the injectivity property holds. However, this cycle packing is not a universal cycle as no viewing window contains all three colours, so the multiset $(1,1,1)$ does not appear. 
\end{example}

\begin{example}
\label{ex:vector:sum:dim:one}
Again take $\mA = \{\ka,\kb, \kc\}$
with window size $m=3$, but now let $f$ be the summation function.
We claim the sequence $(\ka,\ka,\ka,\kc,\kc,\kc,\kb)$ is a cycle packing of size $n=7$. To see this, observe that $f(\ka, \ka, \ka)=0+0+0=0$,
$f(\ka, \ka, \kc)=0+0+2=2$, etc. Continuing these calculations,
the outputted sequence of sums is $(0, 2, 4, 6, 5, 3, 1)$. This has distinct entries, so injectivity is satisfied.
The reader may query the relevance of the summation function.
In fact, this example is a special case of vector sum packings (\cref{def:vector sum packing})
which, in turn, will play a critical role in the construction  
underlying our main theorem. 
\end{example}

Of course, our interest lies in dimension $d>1$,
so we now extend the definition
of cycle packings to higher dimensions.

\begin{definition} \label{def:torus}
Given an alphabet $\mA$,
a size $n$,
a window size $m$,
a dimension $d$, 
and a function $f$ on $\mA^{m^d}$,
an $(\mA, f, n, m, d)$-\textbf{torus packing}
is a function $W : \integers^d \mapsto \mA$
that satisfies
\begin{itemize}
\item {\tt Periodicity.}
For all $\vx \in \integers^d$ and $j \in [0, d-1]$,
we have $W_{\vx} = W_{\vx + n \ve_j}$
where $\ve_j$ denotes the vector
with a $1$ in position $j$ and $0$ elsewhere.
\item {\tt Injectivity.}
For any $\vx$ and $\vy$ in $\integers^d$, if
$f \left( (W_{\vx + \vc})_{\vc \in [0, m-1]^d} \right) =
f \left( (W_{\vy + \vc})_{\vc \in [0, m-1]^d} \right)$,
then $x_j \equiv y_j \bmod n$ for all $j$
(where $x_j$ denote the $j$th coordinate of the vector $\vx$).
\end{itemize}
\end{definition}

Note the key distinction in this definition is that the viewing window shares the same higher dimension as the torus.
Again, by the periodicity property, we can identify a torus packing $W$ with its pattern $(W_{\vx})_{\vx \in [0,n-1]^d}$.
We are now ready to define the central objects of this paper. 
These are grid colourings, a particular case of torus packing
that correspond exactly to the robot position reconstruction problem.

\begin{definition} \label{def:grid:colouring}
An $(n,m,d,k)$-\textbf{grid colouring} is an $(\mA, f, n, m, d)$-torus packing
with $|\mA| = k$
and $f$ mapping any sequence to the multiset of its entries.
\end{definition}


\begin{example}
Consider, again, our example in~\cref{fig:grid:coloring}. 
This is an $(\mA, f, n, m, d)$-torus packing
with $\mA = \{\ka,\kb,\kc\}$, $n=8$, $m=4$, $d=2$ 
and where $f$ is the multiset counting function.
Consequently, since $|\mA| = 3$, this is 
an $(8,4,2,3)$-grid colouring.
\end{example}

The following is our main result on grid colourings,
implemented and tested in Python \cite{ourcode}.
(see also Appendix~B).

\begin{theorem}
\label{th:construction:grid:colouring}
Fix a dimension $d \geq 2$ and a number of colours $k$ of the form $b d + 1$
for some $b \geq 1$.
For any window size $m$ multiple of $2 (k-1)$,
there is an $(n, m, d, k)$-grid colouring $W$
(explicitly constructed in the proof)
with
\[
    n \sim C_k^{1/d}\cdot m^{k-1}
    \qquad
   where
    \qquad
    C_k = \left( \frac{2}{k-1} \right)^{k-1}.
\]
Furthermore, for any $\vx$ in $\integers^d$,
given the multiset of colours in
$(W_{\vx + \vc})_{\vc \in [0, m-1]^d}$
one can compute $\vx$ mod $n$
with $O_{k,d}(1)$ arithmetic operations.
\end{theorem}

Our construction in \cref{th:construction:grid:colouring} of size $n=\Omega_{k,d}(m^{k-1})$
is optimal up to a multiplicative constant.
This fact follows from the following observation,
which is immediate from counting considerations
(injectivity requires the number of possible colour multisets
to be at least the number of windows that they must distinguish).

\begin{observation}
\label{th:bound:size}
The parameters of any $(n, m, d, k)$-grid colouring
satisfy the inequality
\[
    n^d \leq \binom{m^d + k - 1}{k - 1}.
\]
In particular, for fixed dimension $d$
and number of colours $k$,
as the window size $m$ tends to infinity, we have
\[
    n \leq C_k'^{1/d} \cdot m^{k-1} (1 + \bigO(m^{-1}))
   \qquad
   with
    \qquad
    C_k' = \frac{1}{(k-1)!}.
\]
\end{observation}


We conclude this section by discussing related work and
providing technological justifications for our positioning model.

\subsection{Related Work} \label{sec:rel}

The combinatorial structures associated with torus packings (\Cref{def:torus}) have a rich mathematical literature, 
starting from a problem solved in the 19th century, independently rediscovered
by de Bruijn and now known as {\em de Bruijn sequences} (see \cite{de1975acknowledgement}).
The generalisation to higher dimensions was independently considered in several papers
starting from the 1960's (see \cite{macwilliams1976pseudo}) and has developed an extensive
literature (see \eg \cite{hurlbert_bruijn_1993, hurlbert_new_1995, hurlbert_existence_1996})
under the name of {\em de Bruijn tori}. The extension to general combinatorial structures
encoded by sequences as in  \Cref{def:torus}  was proposed  by Chung, Graham and Diaconis \cite{chung_universal_1992},
who considered the one-dimensional problem (``{\em universal cycles}'') for a variety of combinatorial structures.
These early formulations of the problem generally asked for optimal solutions
in which every object in a given combinatorial class is realised exactly once;
in the context of \Cref{def:torus} this corresponds to strengthening injectivity to bijectivity.
However, for practical purposes one can be satisfied with approximately optimal solutions,
and given the difficulty of finding optimal solutions there is also a substantial literature 
(see \eg \cite{curtis_near-universal_2009, blanca_universal_2011, debski_universal_2016})
finding approximate solutions from the perspective of packing (injectivity) and covering (surjectivity).

For the multiset encoding problem considered in this paper, finding the exact optimum is an open problem
even for the one-dimensional setting of universal cycles, considered by Knuth~\cite[Fascicle 3, Section 7.2.1.3, Problem 109]{knuth2013art}.
Hurlbert \etal \cite{hurlbert_universal_2009}
construct universal cycles for multisets with particular parameters,
and Blanca and Godbole~\cite{blanca_universal_2011} consider the problem
when the multisets have bounded multiplicities.
Furthermore, the known results only consider the opposite regime of parameters
from those needed for our application: we consider small palettes of colours (\aka alphabets) and a large window,
whereas previous techniques in the literature only apply to small windows and large alphabets.
This comment also applies to the problem replacing ``multiset'' by ``set'', which is perhaps even more natural mathematically,
given that it can be viewed as a hypergraph version of the Euler tour problem (introduced by Euler in the 18th century).
Following many partial results, an exact solution to this problem (for small windows and large alphabets) was found by
Glock \etal \cite{glock_euler_2020}, solving a conjecture from
Chung \etal~\cite{chung_universal_1992}.
We are unaware of any results in the literature on multiset packing in one dimension
(meaning maximizing the number of multisets that appear in a sequence,
rather than looking for sequences that contain all possible multisets).

Moving on from the abstract mathematical problem, we now consider some 
of the computer science literature
aimed towards the specific application of indoor positioning. 
Much of the early work was surveyed by Burns and Mitchell \cite{burns_coding_1992}.
More recent literature 
(see \eg \cite{bruckstein_simple_2012, berkowitz_robust_2016, makarov_construction_2019,  chee_binary_2019})
has emphasised two further conditions that do not appear
in the mathematical formulation but are naturally desirable for practical implementations: 
namely (a)~computational efficiency, and (b)~robustness against measurement errors.
These works bring a variety of techniques from coding theory to bear on the positioning problem,
providing efficient positioning algorithms with error correction.
However, while these algorithms make natural use of an existing toolkit and are conceptually pleasing,
we will argue that they are not addressing the most appropriate formulation of the problem from the point of view 
of the target application of robot positioning, for which a new paradigm is needed.

\subsection{Motivation and Engineering Aspects} \label{sec:eng}

The focus of this paper is theoretical, concerning near-optimal designs for elegant combinatorial structures, namely torus packings. However, as stated, our motivation derives from the problem of position reconstruction (where the dimension is $2$
and the torus of is considered a square).
The vast number of applications of localization and positioning
have prompted the design of a plethora of systems; see \cite{SurveyIPS19} and references therein.

Given the great practical importance of positioning systems, it behoves us to justify our claim that the approporiate way to model 
them is with multiset counting functions.
We do this in this section by summarising two systems where our code (and variants of it) will be useful (more details in Appendix~A).

\medskip

\noindent \textbf{Light-Based Positioning.}
Recently, visible light based positioning systems
have gained much attention
\cite{VLP17,SurveyIPS19,VLP20} 
for two key reasons: (i)~there is a strong demand for accurate but low-cost solutions,
and (ii)~there have been breakthroughs
in the energy consumption and life expectancy
of light-emitting diodes (LEDs).
Contrary to systems requiring several light emitters
and relying on triangulation
\cite{VLP20},
systems based on universal torus packing
(such as de Bruijn torus
\cite{sinden_sliding_1985,aboufadel_position_2007,schusselbauer2021dothraki}
and our construction)
rely on only one light source,
reducing cost and energy consumption.
Consider a room lit by a light-emitting diode (LED).
A robot moving on the floor wears a light sensor or camera.
A printed film is placed on top of the robot, above this sensor.
The film is printed with a coloured grid,
distorted so that depending
on the position of the robot in the room,
the light coming from the LED
projects a square window
of the coloured grid on the sensor \cite{OGC23}.
If the coloured grid code depends on the respective positions of the colours 
(such as de Bruijn torus),
the light detector must be a camera
\cite{yang_wearables_2015,li2018two}
to recognize the pattern.
If a torus packing for multisets is used instead,
the robot can wear a simple light intensity sensor
to recover the multiset of colours:
the pattern of the colours is not needed,
and we side-step the problematic issue of resolving the image.
Such a device is considerably less expensive
and has much shorter response time (lower latency) than a camera.
It should be noted that the total light intensity
(number of coloured rays) received by the detector
depends on the position and also the orientation of the robot.
Thus, for practical application of the construction presented in this article some redundancy should be added to allow for a unique position decoding in all cases.
Positioning systems based on a de Bruijn torus
also require redundancy to correct errors
\cite{chee_binary_2019}
and allow different orientations of the receiver
\cite{szentandrasi2012uniform}.
We plan to explore redundancy
in universal torus packing for multisets in future work.

\medskip

\noindent \textbf{Ambient Backscatters.}
An ambient backscatter is a small and inexpensive device
that, upon reception of a radio wave, turns it into electricity
and sends back data through a radio signal.
Consider a warehouse where ambient backscatter devices
are regularly placed, forming a grid.
A robot with a radio emitter and receptor needs
to locate itself in the warehouse.
Each backscatter device, when in reach of the 
emitted radio wave,
sends back its identification number (ID).
If all devices have distinct IDs,
the robot can determine its position based
on the signals it receives. 
However, in the interest of energy consumption,
IDs composed of few bits are preferred \cite{Nokia23}.
The same ID can be reused by several devices,
provided that at any position,
the multiset of IDs detected by the robot is unique.
The relative 
position of the 
robot to its neighbouring backscatter devices is unknown, so the IDs should be chosen
following a universal torus packing for multisets
(a de Bruijn torus would require the robot
to have directional antennas,
making the system less practical and more expensive).
For this application, our construction should be adapted to account 
for radial symmetry and decay of signal power in the detection window.

\section{Overview}

We have reduced the position reconstruction problem to
the design of grid colourings.
Consequently our main result, \cref{th:construction:grid:colouring},
is based upon an near-optimal grid colouring construction.
A key building block in this construction will be the following family of cycle packings, which we call vector sum packings.

\begin{definition}
\label{def:vector sum packing}
Given positive integers $n$, $m$, $b$ and $s$, an $(n, m, b, s)$-\textbf{vector sum packing}
is an $(\mA, f, n, m)$-cycle packing with $\mA = [0, s]^b$
and $f(\vz_0, \ldots, \vz_{m-1}) = \vz_0 + \cdots + \vz_{m-1}$.
\end{definition}

\begin{example}
For the special case of {\em vector dimension} $b=1$ we have already encountered
a $(8, 3, 1, 2)$-vector sum packing in \cref{ex:vector:sum:dim:one}.
\end{example}

%

\begin{example}
Let us see an example of a vector sum packing with vector dimension
$b=3$. Set $m=2$, $s=1$, $n=8$ and
$(\vz_0, \vz_1,\dots, \vz_7) = 
\left( \tvect{0}{0}{0} , \tvect{1}{0}{0},\tvect{0}{1}{0},\tvect{0}{0}{0},\tvect{0}{0}{1},\tvect{1}{0}{0},\tvect{0}{1}{1},\tvect{0}{0}{0} \right)$. Then
$f(\vz_0, \vz_1) = \tvect{1}{0}{0}$, 
$f(\vz_1, \vz_2) = \tvect{1}{1}{0}$, 
$f(\vz_2, \vz_3) = \tvect{0}{1}{0}$, 
$f(\vz_3, \vz_4) = \tvect{0}{0}{1}$, 
$f(\vz_4, \vz_5) = \tvect{1}{0}{1}$, 
$f(\vz_5, \vz_6) = \tvect{1}{1}{1}$, 
$f(\vz_6, \vz_7) = \tvect{0}{1}{1}$,
and
$f(\vz_7, \vz_0) = \tvect{0}{0}{0}$.
As these sums all differ, the injectivity property holds and we have a vector sum packing.
\end{example}

The strategy of our construction is encapsulated in the following two key lemmas (proven in \cref{sec:proofs:vector:to:grid} and \cref{sec:proof:vector:colouring}, respectively).
The first reduces the construction of a grid colouring to the construction of a vector sum packing.

\begin{lemma}
\label{th:vector sum packing:to:grid:colouring}
Consider a dimension $d$,
size $n$,
window size $m$,
and number of colours $k$,
and assume the existence of integers $b$ and $s$
satisfying $k = b d + 1$ and $s = \frac{m^{d-1}}{b d}$.
Then the existence of an $(n, m, b, s)$-vector sum packing
implies the existence of a $(n, m, d, k)$-grid colouring.
\end{lemma}

Given \cref{th:vector sum packing:to:grid:colouring}, the final piece in the puzzle is 
a construction of vector sum packings.

\begin{lemma}
\label{th:construction:vector sum packing}
For any $s \geq 1$, $m \geq 2$ and $b \geq 1$,
there exists an $(n_b, m, b, 2s)$-vector sum packing
with
\[
    n_b =
    (2 m s + 1)^{b - 1}
    \left( 2 m s - \frac{1}{s} \right)
    + \frac{1}{s}.
\]
\end{lemma}

To deduce the existence of the construction for
 \cref{th:construction:grid:colouring}, we combine
\cref{th:vector sum packing:to:grid:colouring,th:construction:vector sum packing},
choosing $m$ as a large multiple of $2 b d$,
with $k = b d + 1$ and $s = \frac{m^{d-1}}{2 b d}$.
Then the value of $n_b$ in~\cref{th:construction:vector sum packing}
satisfies $n_b \sim \big( \frac{2}{k-1} \big)^{(k-1)/d}  m^{k-1}$ for large $m$.
We will prove in~\cref{th:decoding:vector:sum:packing} that the position
can be computed with a constant number of arithmetic operations.

We remark that the vector sum packing problem is related to the combinatorial theory
of {\em antimagic labellings} (a natural variant on the classic topic of {\em magic labellings})
in which one is required to label the edges (or vertices) of a graph (or hypergraph)
from a given set of integers (or vectors) so that vertices (or edges) are uniquely
determined by the sum of their incident labels. A well-known conjecture of Ringel
(cited in \cite{hartsfield2013pearls}) states that for any connected graph with $m>1$ edges 
there is an antimagic labelling of its edges by $\{1,\dots,m\}$.
It appears that this connection has not previously been exploited
and may be fruitful for further research.

\section{From Vector Sum Packing to Grid Colouring}
\label{sec:proofs:vector:to:grid}

The aim of this section is to prove \cref{th:vector sum packing:to:grid:colouring},
which reduces the grid colouring problem to the vector sum packing problem.

    \subsection{The Separation Property}

In this section, we consider a fixed dimension $d \geq 2$,
window size $m \geq 2$ and grid size $n \geq m$.
For $\vx \in \integers^d$,
let us define the \emph{window} of corner $\vx$
as the set of points
\[
    \window(\vx) =
    \{\vx + \vc \mid \vc \in [0,m-1]^d\}.
\]
For the grid colouring problem,
we require that the coordinates of each point $\vx$
are uniquely determined modulo $n$ by
the multiset of colours of the points from $\window(\vx)$.
This multiset is denoted by $\multiset(\vx)$,
and the colour of the point $\vx$ is denoted by $\cc(\vx)$.

We now present a sufficient condition, called {\em separation},
that guarantees this property.
Each colour is represented as a pair (pigment, shade).
We divide the set of colours into $d$ pigment classes,
one pigment class $C_i$ for each dimension $i \in [0, d-1]$.
Each pigment class contains $b$ shades.
Thus $C_i = \{c_{i,0}, c_{i,1}, \ldots, c_{i,b-1}\}$
where if $i$ is the green pigment, say, then $c_{i,0}$
represents very light green and $c_{i,b-1}$
represents very dark green.
The idea now is that for each square $\vx \in \integers^d$,
the coordinate $x_i$ modulo $n$ is uniquely determined
by the pigment class $C_i$, for each $i \in [0, d-1]$,
precisely by $\multiset(\vx) \cap C_i$.
In particular, $x_i$ is independent
of any other pigment class $j\neq i$ on $\window(\vx)$,
for example, shades of red.
Evidently, this separation property implies that
if $\vx$ and $\vy$ have the same colour multiset
for every pigment class then $x_i \equiv y_i \bmod n$
for each $i\in [0, d-1]$ and, hence, $\vx \equiv \vy \bmod n$.
Thus the separation property implies
that we have a proper grid colouring.

Rather than working directly with separation,
it will be convenient to consider the following two further conditions 
that together clearly imply separation.

\medskip

\noindent {\tt Dimensional Inconsistency}: Given $\vx$ and $\vy$, if $x_i\neq y_i \bmod n$ then
\[
    \multiset(\vx) \cap C_i \neq \multiset(\vy) \cap C_i.
\]
\noindent {\tt Anti-Dimensional Consistency}: Given $\vx$ and $\vy$, if $x_i = y_i \bmod n$ then
\[
    \multiset(\vx) \cap C_i=\multiset(\vy) \cap C_i.
\]
To ensure anti-dimensional consistency it will be convenient to focus on 
the following condition called {\em quasi-periodicity}, which states that 
if some $\vx$ is coloured a shade of pigment $i$
then any translate $\vy$ of $\vx$ by distance $m$ in any dimension other than $i$
has the same colour as $\vx$
\[
    \forall i \neq j,
    \text{ if } \cc(\vx) \in C_i
    \text{ then } \cc(\vx + m\, \ve_j) = \cc(\vx).
\]

\begin{lemma} \label{lem:quasi-consistency}
Any quasi-periodic grid colouring
satisfies anti-dimensional consistency.
\end{lemma}
\begin{proof}
Take a quasi-periodic colouring.
It suffices to show, given $x_i$,
that $\multiset(\vx) \cap C_i$ is fixed.
This will hold if the number of squares of colour $c_{i,\ell}$
in a window does not change if
we translate by one square in any dimension $j$ other than $i$.
Let us define
\[
    A = \window(\vx) \setminus \window(\vx + \ve_j),
    \qquad
    B = \window(\vx + \ve_j) \setminus \window(\vx).
\]
Since $B$ is obtained from $A$ by translation by $m \ve_j$,
quasi-periodicity implies
that the number of points of colour $c_{i,\ell}$
in $A$ and $B$ are equal.
Thus, the number of points of colour $c_{i,\ell}$
in $\window(\vx)$ and $\window(\vx + \ve_j)$
are equal as well.
As this argument applies for any $j\neq i$, the anti-dimensional consistency property holds. 
Thus, quasi-periodicity implies anti-dimensional consistency.
\end{proof}

    \subsection{Separation via Vector Sum Packing}

Assume we have an $(n,m,b,s)$-vector sum packing
$(\vz_j)_{j \in \integers}$
of size $n$ and window size $m$,
containing vectors in $[0,s]^b$,
where $s$ is a positive integer equal to
$\frac{m^{d-1}}{b d}$ for some dimension $d \geq 2$.
We will now construct a $(n, m, d, k)$-grid colouring
with number of colours $k = b d + 1$.
We detail our algorithm below
and illustrate it in \cref{fig:construction:grid:coloring}.
As we saw in the previous section,
to ensure the separation condition,
it is sufficient that our grid colouring
satisfies dimensional inconsistency and quasi-periodicity.

\begin{itemize}
\item[(a)]
We begin with an initial colouring of the grid
using only $d$ pigments, each coming in $b$ different shades.
Take a pigment $i \in [0, d-1]$ and shade $h \in [0, b-1]$.
Then the point $\vx = (x_0, \ldots, x_{d-1})$ has pigment
$i$ and shade $h$ if and only if
\[
    \sum_{j=0}^{d-1} x_j \equiv i + h d \bmod b d.
\]
This initial colouring is quasi-periodic.
But, of course, it does not satisfy dimensional inconsistency.
To rectify this, we will apply the last unused colour,
which we call {\em blank},
to erase some colours from the initial colouring.
\item[(b)]
Consider a dimension $i \in [0, d-1]$.
We associate to it the pigment $i$
and the set $C_i$ of the corresponding $b$ shades.
For each $j \in \integers$, let
\[
    B_{i,j} =
    \{\vx \mid x_i = j \text{ and } \forall \ell \neq i,\ x_{\ell} \in [0, m-1]\}.
\]
We apply the blank colour
to erase some of the colours from $B_{i,j}$,
so that the number of points of shade $c_{i,\ell}$
is equal to the $\ell$th component of $\vz_j$.
This is always possible, because
in the initial colouring,
$B_{i,j}$ contains $s = \frac{m^{d-1}}{b d}$
occurrences of shade $c_{i,\ell}$,
and the values of the vectors from the vector sum packing
are all in $[0,s]$.
\item[(c)]
We apply quasi-periodicity to
reproduce this construction
on the rest of the grid.
Specifically, point $\vx = (x_0, \ldots, x_{d-1})$
of pigment $i$ and shade $h$
in the initial colouring is erased
if and only if
the point $\vy$ with
$y_i = x_i$ and for all $j \neq i$,
$y_j = x_j \bmod m$
was erased at step (b).
\end{itemize}
The injectivity of the vector sum packing implies
the dimensional inconsistency of our grid colouring.
By construction, our grid colouring is quasi-periodic,
so by \cref{lem:quasi-consistency},
it satisfies the separation property,
concluding the proof
of \cref{th:vector sum packing:to:grid:colouring}.

\subsection{An Example}
\label{sec:example}

\begin{figure}[h!]
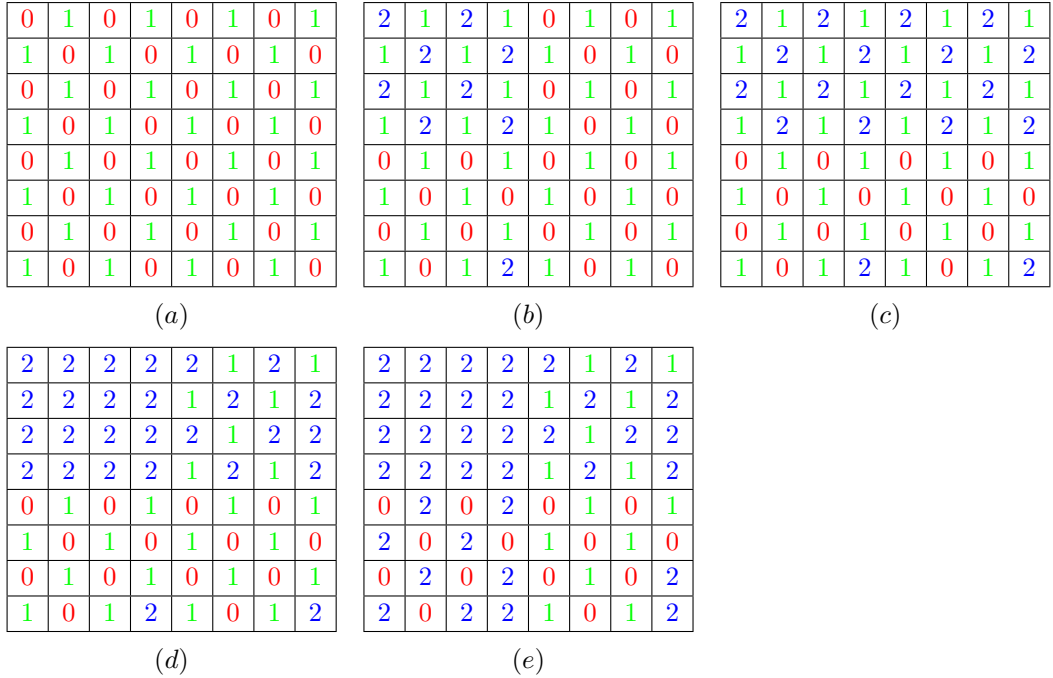

\begin{center}
$\begin{array}{ccc}
\begin{array}{|c|c|c|c|c|c|c|c|}
\hline
\ka & \kb & \ka & \kb & \ka & \kb & \ka & \kb
\\ \hline
\kb & \ka & \kb & \ka & \kb & \ka & \kb & \ka
\\ \hline
\ka & \kb & \ka & \kb & \ka & \kb & \ka & \kb
\\ \hline
\kb & \ka & \kb & \ka & \kb & \ka & \kb & \ka
\\ \hline
\ka & \kb & \ka & \kb & \ka & \kb & \ka & \kb
\\ \hline
\kb & \ka & \kb & \ka & \kb & \ka & \kb & \ka
\\ \hline
\ka & \kb & \ka & \kb & \ka & \kb & \ka & \kb
\\ \hline
\kb & \ka & \kb & \ka & \kb & \ka & \kb & \ka
\\ \hline
\end{array}
&
\begin{array}{|c|c|c|c|c|c|c|c|}
\hline
\kc & \kb & \kc & \kb & \ka & \kb & \ka & \kb
\\ \hline
\kb & \kc & \kb & \kc & \kb & \ka & \kb & \ka
\\ \hline
\kc & \kb & \kc & \kb & \ka & \kb & \ka & \kb
\\ \hline
\kb & \kc & \kb & \kc & \kb & \ka & \kb & \ka
\\ \hline
\ka & \kb & \ka & \kb & \ka & \kb & \ka & \kb
\\ \hline
\kb & \ka & \kb & \ka & \kb & \ka & \kb & \ka
\\ \hline
\ka & \kb & \ka & \kb & \ka & \kb & \ka & \kb
\\ \hline
\kb & \ka & \kb & \kc & \kb & \ka & \kb & \ka
\\ \hline
\end{array}
&
\begin{array}{|c|c|c|c|c|c|c|c|}
\hline
\kc & \kb & \kc & \kb & \kc & \kb & \kc & \kb
\\ \hline
\kb & \kc & \kb & \kc & \kb & \kc & \kb & \kc
\\ \hline
\kc & \kb & \kc & \kb & \kc & \kb & \kc & \kb
\\ \hline
\kb & \kc & \kb & \kc & \kb & \kc & \kb & \kc
\\ \hline
\ka & \kb & \ka & \kb & \ka & \kb & \ka & \kb
\\ \hline
\kb & \ka & \kb & \ka & \kb & \ka & \kb & \ka
\\ \hline
\ka & \kb & \ka & \kb & \ka & \kb & \ka & \kb
\\ \hline
\kb & \ka & \kb & \kc & \kb & \ka & \kb & \kc
\\ \hline
\end{array}
\\[-0.3cm]
\\
(a) & (b) & (c)
\\[5pt]
\begin{array}{|c|c|c|c|c|c|c|c|}
\hline
\kc & \kc & \kc & \kc & \kc & \kb & \kc & \kb
\\ \hline
\kc & \kc & \kc & \kc & \kb & \kc & \kb & \kc
\\ \hline
\kc & \kc & \kc & \kc & \kc & \kb & \kc & \kc
\\ \hline
\kc & \kc & \kc & \kc & \kb & \kc & \kb & \kc
\\ \hline
\ka & \kb & \ka & \kb & \ka & \kb & \ka & \kb
\\ \hline
\kb & \ka & \kb & \ka & \kb & \ka & \kb & \ka
\\ \hline
\ka & \kb & \ka & \kb & \ka & \kb & \ka & \kb
\\ \hline
\kb & \ka & \kb & \kc & \kb & \ka & \kb & \kc
\\ \hline
\end{array}
&
\begin{array}{|c|c|c|c|c|c|c|c|}
\hline
\kc & \kc & \kc & \kc & \kc & \kb & \kc & \kb
\\ \hline
\kc & \kc & \kc & \kc & \kb & \kc & \kb & \kc
\\ \hline
\kc & \kc & \kc & \kc & \kc & \kb & \kc & \kc
\\ \hline
\kc & \kc & \kc & \kc & \kb & \kc & \kb & \kc
\\ \hline
\ka & \kc & \ka & \kc & \ka & \kb & \ka & \kb
\\ \hline
\kc & \ka & \kc & \ka & \kb & \ka & \kb & \ka
\\ \hline
\ka & \kc & \ka & \kc & \ka & \kb & \ka & \kc
\\ \hline
\kc & \ka & \kc & \kc & \kb & \ka & \kb & \kc
\\ \hline
\end{array}
&
\\[-0.3cm]
\\
(d) & (e) & 
\end{array}$
\caption{Illustration of the steps of the proof of \cref{th:vector sum packing:to:grid:colouring}
on a grid of dimension $d = 2$, size $n = 8$, window size $m = 4$,
number of colours $k = 3$.}
\label{fig:construction:grid:coloring}
\end{center}
\end{figure}

An illustration of the proof of \cref{th:vector sum packing:to:grid:colouring}
is given in \cref{fig:construction:grid:coloring}
on a grid of dimension $d = 2$, size $n = 8$, window size $m = 4$,
number of colours $k = 3$. To create the grid-colouring
we use the vector sum packing $(\vz_0, \vz_1, \dots, \vz_7)=(0,0,0,0,2,2,2,1)$
(note each vector has dimension $1$, so is represented by its content).
In particular, $s=2$ and $b=1$; note that $k=bd+1$ and $s=\frac{m^{d-1}}{bd}$.
Recall, the aim is that number of occurrences of $\ka$ (\resp $\kb$) in a $4$ by $4$ square characterizes its row (\resp column) number.
Following the proof of \cref{th:vector sum packing:to:grid:colouring}, we
start with a periodic colouring, represented in $(a)$.
We now use the blank colour $\kc$ to erase some of the $\ka$.
First, in $(b)$, we erase entries in the first $4$ columns
so the number of occurrences of $\ka$ in these columns 
for the eight rows are $(0,0,0,0,2,2,2,1)$.
Second in $(c)$, we apply quasi-periodicity on the rest of the grid.
Next, in $(d)$ and $(e)$, we apply the same approach to  
erase some of the $\kb$. The final result is $(e)$.
No two $4 \times 4$ subsquares contain the same multiset of colours.

Let us now illustrate how this grid is used for localization.
Recall that our convention is
to number the rows from top to botton,
and column from left to right,
both starting at $0$.
Assume we are measuring in a window (\ie a $4 \times 4$ subsquare)
the multiset of colors containing
$5$ occurrences of $\ka$,
$3$ occurrences of $\kb$
and $8$ occurrences of $\kc$.
We wish to locate this window in the grid (e).
The naive algorithm is to consider each possible window in the grid
and compare the multisets of colors.
This becomes costly for large grids,
so we present a more efficient algorithm.
By convention, color $\ka$ is used to determine the row.
Looking at the vector sum packing $(0,0,0,0,2,2,2,1)$,
we observe that the sequence whose $j$th element
is the sum of $m=4$ consecutive elements starting at position $j$,
is $(0,2,4,6,7,5,3,1)$.
The number $5$ is located at position $5$ in this sequence,
so the upper-left corner of the window we are seeking
has row number $5$.
To determine the column, we consider the color $\kb$.
Its number of occurrence $3$ has position $6$
in 
$(0,2,4,6,7,5,3,1)$,
so the column number is $6$.
On the torus (e), the $4 \times 4$ subsquare with top left corner
in row $5$ and column $6$ indeed contains
$5$ occurrences of $\ka$,
$3$ occurrences of $\kb$
and $8$ occurrences of $\kc$.

Observe that for any fixed dimension,
the localization problem in the grid
reduces in constant complexity
to the problem of computing the position of a vector
in a vector sum packing.
We call this second problem \emph{decoding}.
In the next section, we will present our construction
for vector sum packings,
as well as a decoding algorithm
with constant complexity
(in the number of arithmetic operations,
as the dimension $d$ and parameter $b$, defined there,
are fixed).

A grid colouring of size $256$,
window size $8$,
and $5$ colours is presented in the appendix.

\section{Vector Sum Packing} \label{sec:proof:vector:colouring}

This section presents our construction of vector sum packings,
thus proving \cref{th:construction:vector sum packing},
which is the last missing ingredient for the proof of our main result
\cref{th:construction:grid:colouring}.

\subsection{Profiles and Duals} \label{sec:profiles}

To build vector sum packings (see \cref{def:vector sum packing}),
we introduce certain integer sequences that we call \emph{profiles}.
Profiles with different parameters will be used
to fill the coordinates of the vector sum packing,
in \cref{th:profile:to:vector:colouring}.

The \emph{$m$-dual} of a profile $w$
is an integer sequence of same length $|w|$,
defined as the sum of the elements of $w$
on a cycling window of length $m$
\[
    \dual_m(w) =
    \bigg(
        \sum_{j=i}^{i+m-1} w_{i \bmod{|w|}}
    \bigg)_{i \in [0, |w|-1]}.
\]

Consider integers $s \geq 1$ and $m \geq 2$.
Let us write sequences of length $1$ as $( \sigma )$
and sequences of length $m$ as $(\sigma_0, \sigma_1, \ldots, \sigma_{m-1})$.
For a finite sequence $L$ and a nonnegative integer $T$,
the sequence obtained by concatenating $T$ copies of $L$
one after the other is denoted by $L^T$.
Let $\emptyset$ denote the emptysequence, and let $a \cdot b$
denote the concatenation of the sequences $a$ and $b$,
so $(1,2,3) \cdot (4)$ is equal to $(1,2,3,4)$.
In the following tables, the rows and columns
are numbered starting at $0$, in {\color{red} red}.
Straight lines have been added between some of the rows
to distinguish parts of the tables
following different rules.

Let us define the sequence $\profile(s,m,0)$ as the concatenation
of the cells from the following table,
read line by line from top left to bottom right.
\[
\resizebox{\textwidth}{!}{
$\begin{array}{c|ccccc}
& \color{red} 0 & \color{red} 1 & \color{red} 2 & \color{red} \cdots & \color{red} s-1
\\ \hline \\[-10pt]
\color{red} 0
& (0,\dots,0)
& (2,\dots,2)
& (4,\dots,4)
& \cdots
& (2s-2,\dots,2s-2)
\\ \hline \\[-10pt]
\color{red} 1
& (2s,\dots,2s,2s-1)
& (2s-2,\dots,2s-2,2s-3)
& (2s-4,\dots,2s-4,2s-5)
& \cdots
& (2,\dots,2,1)
\end{array}$}
\]
For example, we have
$\profile(1,3,0) = (0,0,0,2,2,1)$
and $\profile(2, 2, 0) = (0,0,2,2,4,3,2,1)$.

\begin{lemma}
\label{lem:profile:zero} 
$\profile(s,m,0)$ has length $2ms$.
Furthermore, its $m$-dual is
\[
    \dual_m(\profile(s,m,0)) =
    (0, 2, 4, \dots, 2ms-2, 2ms-1, 2ms-3, 2ms-5, \ldots, 1).
\]
\end{lemma}
\begin{proof}
$\profile(s,m,0)$ has length $2ms$ because each entry in the table is a sequence of cardinality $m$ and there
are $2s$ entries in the table. The reader may easily verify that the $m$-dual begins with non-negative even numbers increasing up to $2ms-2$ followed by positive odd numbers 
decreasing down from $2ms-2$. Thus, the $m$-dual contains every integer in $[0, 2m-1]$ exactly once.
\end{proof}

Let us define the $(s,m,0)$-decoding function
as the function that associates
to an integer $v$ its smallest index in the $m$-dual of $\profile(s,m,0)$.
For example, the $(2,2,0)$-decoding function sends $6$ to $3$,
because the $2$-dual of $\profile(2, 2, 0) = (0,0,2,2,4,3,2,1)$
is $(0,2,4,6,7,5,3,1)$ ,
where the first (and only) occurrence of $6$ is at position $3$.

\begin{corollary}
\label{th:decode:profile:zero}
The $(s,m,0)$-decoding function is
\[
    v \mapsto
    \begin{cases}
        \frac{v}{2} & \text{if $v$ is even},
        \\
        2 m s - 1 - \frac{v-1}{2} & \text{if $v$ is odd}.
    \end{cases}
\]
It is computable in a constant number of arithmetic operations.
\end{corollary}

For any positive integer $T$ and $m \geq 2$,
let us define the sequence $\profile(s,m,T)$
as the concatenation of the cells from the following table.
\[
\resizebox{\textwidth}{!}{
$\begin{array}{c|ccccc}
& \color{red} 0 & \color{red} 1 & \color{red} 2 & \color{red} \cdots & \color{red} s-1
\\ \hline \\[-10pt]
\color{red} 0
& (0,\dots,0,0)^T
& (0,\dots,0,2)^T
& (0,\dots,0,4)^T
& \cdots
& (0,\dots,0,2s-2)^T
\\
\color{red} 1
& (0,\dots,0,0,2s)^T
& (0,\dots,0,2,2s)^T
& (0,\dots,0,4,2s)^T
& \cdots
& (0,\dots,0,2s-2,2s)^T
\\
\color{red} \vdots & \vdots &\vdots &  &  & \vdots
\\
\color{red} m-1
& (0,2s,\dots,2s,2s)^T
& (2,2s,\dots,2s,2s)^T
& \cdots
& \cdots
& (2s-2,2s,\dots,2s)^T
\\ \hline \\[-10pt]
\color{red} m
& (2s)^{mT-1}
& (2s-1, 2s, \dots, 2s)^T
& (2s-3, 2s, \dots, 2s)^T
& \cdots
& (3, 2s, \dots, 2s) ^T
\\ \hline \\[-10pt]
\color{red} m+1
& (1, 2s, 2s, \dots, 2s)^{T-1}
& (1,2s,\dots,2s,2s-2)^T
& (1,2s,\dots,2s,2s-4)^T
& \cdots
& (1,2s,\dots,2s,2)^T
\\ \hline \\[-10pt]
\color{red} m+2
& (1,2s,\dots,2s,0)^T
& (1,2s,\dots,2s-2,0)^T
& (1,2s,\dots, 2s ,2s-4,0)^T
& \cdots
& (1,2s,\dots, 2s,2,0)^T
\\
\color{red} \vdots & \vdots &\vdots &  &  & \vdots
\\
\color{red} 2m-1
& (1,2s,0,\dots,0)^T
& (1,2s-2,0,\dots,0)^T
& \cdots
& \cdots
& (1,2,0,\dots,0)^T
\\ \hline \\[-10pt]
\color{red} 2m
& (1,0,\dots,0)^{T-1} \cdot (1)
& \emptyset
& \emptyset
& \cdots
& \emptyset
\end{array}$}
\]

\begin{lemma}
\label{lem:profile:T}
For any $m \geq 2$,
$\profile(s,m,T)$ has length $m ((2 m s + 1) T - 2)$.
Furthermore, its $m$-dual is obtained by concatenation
of the cells of the following table
\[
\resizebox{\textwidth}{!}{
$\begin{array}{c|ccccc}
& \color{red} 0 & \color{red} 1 & \color{red} 2 & \color{red} \cdots & \color{red} s-1
\\ \hline \\[-10pt]
\color{red} 0
& (0)^{m T}
& (2)^{m T}
& (4)^{m T}
& \cdots
& (2s-2)^{m T}
\\ \hline \\[-10pt]
\color{red} 1
& (2s)^{m T - 1}
& (2s+2)^{m T}
& (2s+4)^{m T}
& \cdots
& (4s-2)^{m T}
\\
\color{red} 2
& (4s)^{m T - 1}
& (4s+2)^{m T}
& (4s+4)^{m T}
& \cdots
& (6s-2)^{m T}
\\
\color{red} \vdots & \vdots & \vdots & \vdots && \vdots
\\
\color{red} m-1
& (2(m-1)s)^{m T-1}
& (2(m-1)s+2)^{m T}
& (2(m-1)s+4)^{m T}
& \cdots
& (2ms -2)^{m T}
\\ \hline \\[-10pt]
\color{red} m
& (2ms)^{mT-1}
& (2 m s-1)^{m T}
& (2 m s - 3)^{m T}
& \cdots
& (2 (m-1) s + 3)^{m T}
\\ \hline \\[-10pt]
\color{red} m+1
& (2 (m-1) s + 1)^{m T - 1}
& (2 (m-1) s - 1)^{m T}
& (2 (m-1) s - 3)^{m T}
& \cdots
& (2 (m-2) s + 3)^{m T}
\\
\color{red} \vdots & \vdots & \vdots & \vdots && \vdots
\\
\color{red} 2m-1
& (2 s + 1)^{m T - 1}
& (2 s - 1)^{m T}
& (2 s - 3)^{m T}
& \cdots
& (3)^{m T}
\\ \hline \\[-10pt]
\color{red} 2m
& (1)^{m T - 1}
& \emptyset
& \emptyset
& \cdots
& \emptyset
\end{array}$}
\]
\end{lemma}

\begin{proof}
In the table, there are $2 m s - 2$ cells containing sequences of length $m T$,
one cell containing a sequence of length $m T - 1$,
one cell containing a sequence of length $m (T - 1)$
and one celle containing a sequence of length $m (T - 1) + 1$,
so
$\profile(s,m,T)$ has length
\[
    (2 m s - 2) m T + m T - 1 + m (T - 1) + m (T - 1) + 1
    =
    m T (2 m s + 1) - 2 m
\]
as desired.
Again, the reader may verify that the $m$-dual begins with non-negative even numbers 
increasing up to $2ms$ followed by positive odd numbers 
decreasing down from $2ms-1$ (repeated in the quantities specified).
\end{proof}

We define the $(s,m,T)$-decoding function
as the function that associates
to an integer $v$ its smallest index in the $m$-dual of $\profile(s,m,T)$.

\begin{corollary}
\label{th:decode:profile}
The $(s,m,T)$-decoding function output on the input $v$
is computed using the following algorithm.
If $v$ is even, we define
$r = \lfloor \frac{v}{2 s} \rfloor$
and $c =  \frac{v}{2} \bmod [s]$.
They represent the row and column in the $m$-dual from \cref{lem:profile:T}.
Then the output of the decoding function is
\[
    (r\, m\, T\, s - r + 1) \indic_{r > 0}
    + (c\, m\, T - 1 + \indic_{r = 0}) \indic_{c > 0}.
\]
Otherwise, $v$ is odd and we define
$r = \lfloor \frac{2\, m\, s - v + 1}{2 s} \rfloor$
and $c = \frac{2\, m\, s - v + 1}{2} \bmod s$.
The output is then
\[
    1 + m (m\, T\, s - 1)
    + (m\, T\, s\, r - r) \indic_{r > 0}
    + (c\, m\, T\, - 1) \indic_{c > 0}.
\]
It is computable in a constant number of arithmetic operations.
\end{corollary}

\subsection{Generating a Vector Sum Packing} \label{sec:generating:vector:sum:packing}

We will explicitly construct
the vector sum packing by combining profiles,
thus proving \cref{th:construction:vector sum packing}.

Recall that an $(n,m,b,s)$-vector sum packing
is characterised by its pattern
$(\vz_0, \dots, \vz_{n-1})$
where each $\vz_i$ is a vector in $[0,s]^b$.
These vectors will be defined as
the columns of a matrix $M^{(b)}$.
Furthermore, the rows of this matrix will be constructed using sequences
given by the $\profile(s,m,T)$.

The matrix $M^{(b)}$ will have $b$ rows and $m\cdot T_b=n_b$
columns. Specifically, given $s \geq 1$ and $m \geq 2$,
we set $T_0 = 0$ and $T_1 = 2s$. Then, for each $b \geq 1$,
we recursively set
\[
    T_{b+1} = (2 m s + 1) T_b - 2.
\]
Thus we obtain the dimensions of our $M^{(b)}$ matrices.
To fill in the entries of the matrices we again apply a
recursive construction.

$\bullet$ For $b=1$, the matrix $M^{(1)}$ has only one row, which is identical to the sequence $\profile(s,m,0)$.
We remark that $M^{(1)}$ does indeed have $n_1=m\cdot T_1= 2ms$
columns, as required by \cref{lem:profile:zero}.

$\bullet$ Next consider the case $b \geq 2$.
The basic idea is that $M^{(b)}$ should simply be the concatenation on $2ms+1$ copies of $M^{(b-1)}$, 
plus an additional row identical to the sequence $\profile(s,m,T_{b-1})$,
which will be used to distinguish between the different copies.

However, this basic idea does not scale correctly,
so instead of concatenating identical copies of $M^{(b-1)}$, we
also concatenate truncated copies of $M^{(b-1)}$.
Specifically, we allow for the truncated matrix $M^{(b-1, \star)}$ which is identical to $M^{(b-1)}$
except that its first column is removed.

We now set $M^{(b)}= M^{(b-1, 0)}\circ M^{(b-1, 2)}\circ  \cdots \circ M^{(b-1, 2 m s)} \circ M^{(b-1, 2 m s-1)}\circ\cdots \circ M^{(b-1, 1)})$, where
each $M^{(b-1, \ell)}$ is either $M^{(b-1)}$ or $M^{(b-1,\star)}$ and where $\circ$ denotes the concatenation operation.
Thus, it remains, to prescribe, for each 
for $\ell \in [0, 2 m s]$ whether 
$M^{(b-1, \ell)}$ is set equal to $M^{(b-1)}$ or $M^{(b-1,\star)}$.
To do this, we use the $m$-dual of the $\profile(s,m,T_{b-1})$.
In particular, define $I_{\ell}$ to be the set of indices
where the $m$-dual of the profile takes value $\ell$.
That is,
\begin{equation}
\label{eq:def:Iell}
    I_{\ell} = \{i \mid \dual_m(\profile(s,m,T_{b-1}))_i = \ell\}
\end{equation}
Recall that $n_b= mT_b = m\cdot ((2ms+1)\cdot T_{b-1}-2)$.
Observe then, from \cref{lem:profile:T}, that
the $(I_{\ell})$s are disjoint integer intervals
of length either $n_{b-1}$ or $n_{b-1} - 1$
whose union is $[0, n_b - 1]$. 
For each $\ell \in [0, 2 m s]$, we now define
\[
    M^{(b-1, \ell)} =
    \begin{cases}
    M^{(b-1)} & \text{if } |I_{\ell}| = n_{b-1},\\
    M^{(b-1, \star)} & \text{if } |I_{\ell}| = n_{b-1} - 1.
    \end{cases}
\]
The resultant construction of $M^{(b)}$ is then illustrated 
in \cref{fig:M:b}.
For example, for $s=1$, $m=2$ and $b=2$, we have
$T_0 = 0$ and $T_1 = 2$, so
\begin{align*}
    \profile(s,m,T_{b-1}) &=
    \profile(1,2,2) =
    (0, 0, 0, 0, 0, 2, 0, 2, 2, 2, 2, 1, 2, 1, 0, 1),
    \\
    \dual_m(\profile(s,m,T_{b-1})) &=
    (0, 0, 0, 0, 2, 2, 2, 4, 4, 4, 3, 3, 3, 1, 1, 1),
    \\
    \profile(s,m,T_{b-2}) &=
    \profile(1, 2, 0) =
    (0,0,2,1),
\end{align*}
so
\[
    M^{(2)} =
    \left(
    \begin{array}{cccccccccccccccc}
    0 & 0 & 2 & 1 & 0 & 2 & 1 & 0 & 2 & 1 & 0 & 2 & 1 & 0 & 2 & 1\\
    0 & 0 & 0 & 0 & 0 & 2 & 0 & 2 & 2 & 2 & 2 & 1 & 2 & 1 & 0 & 1
    \end{array}
    \right).
\]

\begin{figure}[h]
\[
\def\cellwidth{1.72cm}
\def\spacea{\hphantom{aaaaaaaa}}
\vspace{-0.1cm}
    M^{(b)} =
\begin{array}{c}
    \begin{array}{C{\cellwidth}C{\cellwidth}C{\cellwidth}C{\cellwidth}C{\cellwidth}C{\cellwidth}}
        $\overbrace{\spacea}^{I_0}$ &
        $\overbrace{\spacea}^{I_2}$ &
        &
        $\overbrace{\spacea}^{I_{2 m s}}$ &
        &
        $\overbrace{\spacea}^{I_1}$
    \end{array}
    \\
    \begin{array}{|C{\cellwidth}|C{\cellwidth}|C{\cellwidth}|C{\cellwidth}|C{\cellwidth}|C{\cellwidth}|}
        \hline
        &&&&& \\
        $M^{(b-1,0)}$ &
        $M^{(b-1,2)}$ &
        $\cdots$ &
        $M^{(b-1, 2 m s)}$ &
        $\cdots$ &
        $M^{(b-1, 1)}$
        \\
        &&&&& \\
        \hline
        \multicolumn{6}{|c|}{\profile(s,m,T_{b-1})} \\
        \hline
    \end{array}
\end{array}
\]
\caption{The recursive construction of the matrix $M^{(b)}$.
The indices $j$ for the sets $I_j$ on top of the figure
are first increasing even numbers, then decreasing odd numbers.}
\label{fig:M:b}
\end{figure}

As stated, we will take our vectors $(\vz_0, \vz_1, \dots, \vz_{n_b-1})$
to be the columns of $M^{(b)}$.
That is, let $\vv_{i,j} = M^{(b)}_{i \bmod n_b,\, j}$.
Then, for any $i \in \integers$, we have $\vz_i = (\vv_{i,j})_{j \in [0,b-1]}$.

\begin{lemma}
\label{th:profile:to:vector:colouring}
The sequence of vectors $(\vz_0, \vz_1, \dots, \vz_{n_b-1})$
is the pattern of an $(n_b,m,b,2s)$-vector sum packing.
\end{lemma}
\begin{proof}
Given $(\vz_0, \vz_1, \dots, \vz_{n_b-1})$
recall that $f(\vz_j, \vz_{j+1},\dots, \vz_{j+m-1}) = \vz_j+ \vz_{j+1}+\cdots +\vz_{j+m-1}$, 
where the indices are taken modulo $n_b$. Our task is to prove that $f$ is injective.

We will show this by extending the definition of an $m$-dual to matrices: 
we define the $m$-dual of a matrix $M$ 
with $n$ columns to be the matrix $D$ of the same dimensions
where for all $i$, the $i$th column of $D$
is equal to the sum of the columns of $M$
of all indices in $[i, i+m-1]$ modulo $n$.

Proving $(\vz_0, \vz_1, \dots, \vz_{n_b-1})$
is an $(n_b,m,b,2s)$-vector sum packing 
is then equivalent to proving that the $m$-dual of $M^{(b)}$
does not contain any identical columns.
We will do so by induction on $b$.

\proofsubparagraph*{Initialization.}
For the base case $b=1$, recall that $M^{(b)} = \profile(s,m,0)$.
Hence, by \cref{lem:profile:zero},
$\dual_m(M^{(b)})$ does not contain two identical columns.

\proofsubparagraph*{Induction Step.}
For the induction hypothesis, assume that $\dual_m(M^{(b-1)})$
does not contain two identical columns.
Now define $\profile^*(s,m,T)$ to be equal to $\profile(s,m,T)$
except with the first $0$ removed.
By construction, each row $j \in [0, b-1]$
of $M^{(b)}$ is a concatenation of the form
$P_0\circ P_1 \circ \cdots \circ P_r$
where each $P_i$ is equal either to
$\profile(s,m,T_j)$ or to $\profile^*(s,m,T_j)$.
Next observe that both $\profile(s,m,T_j)$ and $\profile^*(s,m,T_j)$
start with $m-1$ occurrences of $0$.
This implies that $m$-dual and concatenation commutes as
\[
    \dual_m(P_0\circ P_1 \circ \cdots \circ P_r) =
    \dual_m(P_0)\circ\dual_m(P_1)\circ \cdots \circ\dual_m(P_r).
\]
Therefore the $m$-dual of the matrix $M^{(b)}$ is
\begin{equation}
\label{eq:dual:M}
\def\cellwidth{2.9cm}
\def\spacea{\hphantom{aaaaaaaaaaaaaaa}}
    \dual_m(M^{(b)}) =
    \begin{array}{c}
        \begin{array}{C{\cellwidth}C{\cellwidth}C{\cellwidth}}
            $\overbrace{\spacea}^{I_0}$ &
            &
            $\overbrace{\spacea}^{I_1}$
        \end{array}
        \\
        \begin{array}{|C{\cellwidth}|C{\cellwidth}|C{\cellwidth}|}
            \hline
            && \\
            $\dual_m(M^{(b-1,0)})$ &
            $\cdots$ &
            $\dual_m(M^{(b-1, 1)})$
            \\
            && \\
            \hline
            \multicolumn{3}{|c|}{\dual_m(\profile(s,m,T_{b-1}))} \\
            \hline
        \end{array}
    \end{array}
\end{equation}

Assume that two columns $c$ and $c'$ with column indices $i$ and $i'$
in $\dual_m(M^{(b)})$ are equal.
Let $\ell \in [0, 2 m s]$ denote the value of $c$ in its final row $b-1$. 
Thus, $c'$ also has entry $\ell$ in its last row.
But the last row $b-1$ of $\dual_m(M^{(b)})$
is defined to equal $\dual_m(\profile(s,m,T_{b-1}))$.
Hence, by definition of $I_\ell$ from \cref{eq:def:Iell},
we have $i \in I_{\ell}$ and $i' \in I_{\ell}$.
Let $d$ and $d'$ be obtained from $c$ and $c'$
by removing their last row $b-1$.
By \cref{eq:dual:M},
both $d$ and $d'$ belong to $\dual_m(M^{(b-1,\ell)})$.
Thus, they both belong to $\dual_m(M^{(b-1)})$. 
By the induction hypothesis,
this implies $i = i'$ and concludes the proof.
\end{proof}

Let the $(n,m,b,s)$-decoding function for vector sum packing
be defined as the function that
inputs an integer vector $\vx$
and outputs its index
in the $m$-dual of the $(n,m,b,s)$-vector sum packing we defined.

\begin{corollary}
\label{th:decoding:vector:sum:packing}
For fixed $b$, the $(n_b,m,b,2s)$-decoding function
for our vector sum packing
is computable in a constant number of arithmetic operations.
\end{corollary}

\begin{proof}
This decoding function is computed recursively,
following the recursive construction of matrix $\dual_m(M^{(b)})$
from \cref{eq:dual:M}.
It inputs a vector $\vx$ and an auxiliary Boolean parameter $B$,
initialized at \emph{False},
that indicates if we are looking for localization
in a vector sum packing where the first vector has been removed.

We first use the decoding function for profiles
from \cref{th:decode:profile:zero,th:decode:profile}
on the last coordinate of $\vx$
to compute an integer $c$,
and set $p = c$ if $B$ is equal to \emph{False},
and $p = c-1$ if $B$ is equal to \emph{True}.
For $b = 1$, as $\vx$ is a vector of dimension $1$,
the algorithm stops and $p$ is returned.
Otherwise, by construction,
$p$ is the smallest possible index
for any vector whose last coordinate is equal
to the last coordinate of $\vx$.
Now, in \cref{eq:dual:M},
we want to determine whether the matrix
$\dual_m(M^{b-1,j})$ on top of $p$
is equal to $\dual_m(M^{b-1})$
or to $\dual_m(M^{b-1,\star})$.
As explained in the construction,
this is decided by looking at $\dual_m(\profile(s,m,T_{b-1}))$
from \cref{lem:profile:T}.
In this sequence, let $\ell$ denote
the number of repetition of the element at position $p$.
\begin{itemize}
\item
If $\ell = m T_{b-1}$,
then we are working with $\dual_m(M^{b-1})$.
In that case,
we call recursively the $(n_{b-1},m,b-1,2s)$-decoding function
on the vector $\vx$ without its last coordinate,
with auxiliary parameter equal to \emph{False}.
The output is added to $p$ and returned.
\item
Otherwise, we have $\ell = m T_{b-1} - 1$,
and we are working with $\dual_m(M^{b-1,\star})$.
We call recursively the $(n_{b-1},m,b-1,2s)$-decoding function
on the vector $\vx$ without its last coordinate,
with auxiliary parameter equal to \emph{True}.
The output is added to $p$ and returned.
\end{itemize}
Deciding whether $\ell = m T_{b-1}$
or $\ell = m T_{b-1} - 1$
is achieved by looking at the parity of $p$
and its value modulo $s$.
This recursive construction has depth $b$,
so for $b$ fixed, it requires only
a constant number of arithmetic operations.
Python code for this algorithm is provided in \cite{ourcode}.
\end{proof}

\subsection{Summary and Complexity}

Let us summarize our algorithm constructing a grid coloring.
It inputs a dimension $d \geq 2$
and two other parameters $b \geq 1$ and $t \geq 1$.
The first step is to compute the window size $m = 2bdt$,
the number of colors $k = bd + 1$,
the parameter $s = (2 b d)^{d-2} t^{d-1}$
to ensure $2 s = \frac{m^{d-1}}{b d}$,
and the grid size $n = (2ms + 1)^{b-1} (2ms - 1/s) + 1/s$.
The second step is to construct an $(n,m,b,2s)$-vector sum packing,
as described in \cref{sec:generating:vector:sum:packing},
using the \emph{profile} sequences
defined in \cref{sec:profiles}.
In the third step, we finally colour the points
of our grid of side $n$ and dimension $d$.
Our set of colours is $\{0, 1, \ldots, k-1\}$. Here the 
last colour $k-1$ is a ``\emph{blank}'' color
used to erase the other colours.
The other colours are divided into $d$ sets,
called pigment classes.
Each pigment class contains $b$ colours,
which we call its \emph{shades}.
To each dimension of the grid is associated a unique pigment.
The vector sum packing is then used,
as explained in \cref{sec:proofs:vector:to:grid},
to colour the points of the grid.

Next consider the localization problem.
Localization means, given a multiset $S$ of colours,
the recovery of the unique window of size $m$ in the grid
that contains this multiset of colors
(if it exists).
To achieve localization, we proceed dimension by dimension.
We count in $S$ the colours from the pigment class
corresponding to the dimension considered
and make a vector out of it.
For example, if the dimension corresponds to the colours ${4, 5, 6}$
and $S$ contains three occurrences of the color $4$, zero occurrences of the colour $5$
and two occurrences of the colour $6$,
the vector is $(3,0,2)$.
We use the decoding algorithm for vector sum packing
from \cref{th:decoding:vector:sum:packing}
to translate this vector into a coordinate.
Having achieved this for every coordinate,
we deduce the position of the window
whose multiset of colours is equal to $S$.

The construction of the grid and localization procedure
are illustrated in \cref{sec:example}.
We measure complexity as the number of arithmetic operations
for $b$ and $d$ fixed, while $t$ goes to infinity.
The construction of the grid has complexity
proportional to the size of its output,
which is $\bigO(n^d) = \bigO(t^{b\, d^2})$.
\cref{th:decoding:vector:sum:packing}
implies that the complexity of the localization algorithm is constant.

\section{Conclusion}

Many interesting directions for future research remain, both
theoretical and practical. One nice extension would be
to make our grid colouring robust
by allowing for error detection and correction.
This has been achieved for other cycle and torus packing problems
(see, for example, \cite{bruckstein_simple_2012, berkowitz_robust_2016, makarov_construction_2019,  chee_binary_2019})
and is a necessary step for practical applications.
Another valuable contribution
would be to study disc-like windows rather than
the square windows examined in this article. This would match
the natural shape of the domain
where an emission emitted at a point is detectable.


\section*{Acknowledgements}
Part of the work was been carried out at the Laboratory for Information, Networking and Communication Sciences (www.lincs.fr). 
We thank Dr. Paolo Baracca, Siu-Wai Ho, Kenneth Shum and many colleagues for their great support and discussions.

\bibliography{biblio}

\appendix

This appendix provides additional details
on the practical applications
of our mathematical construction
and on our Python implementation,
including a grid colouring of size $256$,
window size $8$ and number of colours $5$.

\section{Engineering Implementation Details}

This appendix presents some implementation details,
providing further evidence of
the applicability of universal torus packing for multisets
for positioning (see also \cref{sec:eng}).

\subsection{Light-Based Positioning} 

Over recent years, visible light based indoor localization systems \cite{VLP17} have gained considerable attention for two key reasons: (i) strong demand for highly accurate but low-cost solutions and (ii) the global emergence of light-emitting diode (LED) based lighting in replacing fluorescent or traditional lamps. LED lighting is much more energy efficient and has longer life expectancy. LED lighting is by default for illumination but we can reuse it for localization (dual use) or even for data transmission. Given the proliferation of LED lighting and also the emerging development of visible light communications \cite{Haas15}, visible light based localization/positioning is ready for commercial deployment and system development \cite{Till21}.

In comparison to existing schemes \cite{VLP20}, which usually require either multiple light sources similar to the requirement of Global Positioning System (GPS) or multiple optical receivers/photodiodes (PDs) for triangulation or trilateration methods, 
a novel light based localization method can use only one LED and one light sensor (at the user device) \cite{OGC23} 
with a suitable grid colour 
mapping 
mechanism 
or conversion,  
which allows low system cost and high positioning accuracy, capable of cm-resolution or even higher \cite{Nature23}. Note that since common PD has response time at about 25 ms or even faster, one can expect low latency in the localization, which is particularly relevant for unmanned ground vehicle application and mobile users.

In the following, we explain its design and mechanism for the completeness of information. There is one LED. Without loss of generality, we assume the LED is a square light emitting plate with width equal to $W$ millimeters (mm). The user device has a PD which is a point source and equipped with a grid colour-coded film. The film is printed with a set of coloured dots, which is to filter light shined on it. Note that the projection of the above LED on the film when being captured by the PD is also a square (we call it the \emph{projected area} or \emph{window}) and will reach the light sensitive area of the PD accordingly. For example, the sensitive area of a conventional PD can be 0.6 mm $\times$ 0.6 mm \cite{Sensor}, which is a tiny point.

\begin{figure}[h!] 
    \begin{center}
    \includegraphics[width=0.6\columnwidth]{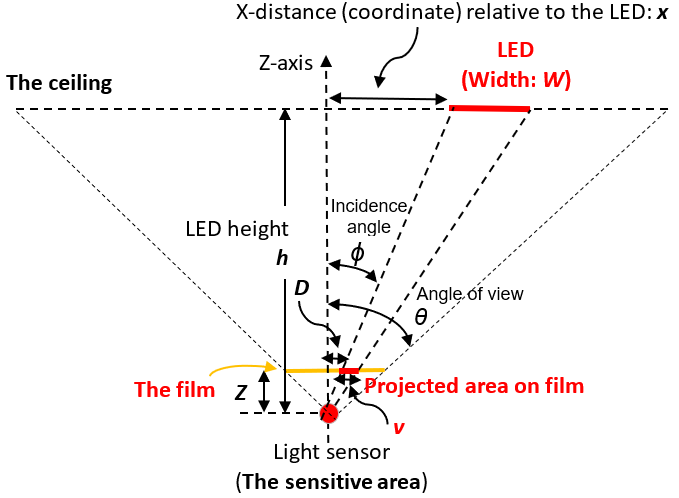} 
    \caption{The system schematic and localization method (Side View) \cite{OGC23}. The width of the projected area is denoted by $v$. Note that the picture is not to scale.} 
    \label{fig:Location}
    \end{center}
\end{figure} 

As shown in \cref{fig:Location}, when the LED light shines on the film, a projected area (window) on the film is produced. Note that the projected area would appear at different location on the film, depending on the position of the PD (\ie, the location of the user device) relative to that of the LED or the light incidence angle form the LED.
\cref{fig:Location} explains the geometry and also how to recover the relative coordinate of the user compared to the position of the LED \cite{OGC23}.
Assume that the film is parallel to the ceiling where the LED is installed. 
We denote the width of the projected area by $v$. 
By similar triangles, we have $v/W = Z/h$,
where $h$ and $Z$ are the heights, such that $v = (W \times Z)/h$. For example, given a LED with $W = 45$ mm, $h = 2$ meters and $Z = 10$ mm due to the film, it is easy to determine that $v = 0.225$ mm.

For user localization, we would colour the dots on the film such that for each projected area on the film, the PD through the film would receive distinct light colour intensity multisets, given by the totality of the light arriving at the PD through the dots of the projected area. Note that the PD does not need to detect the exact pattern of the coloured dots or measure the light of individual dot: the requirement in colouring the dots is that the sum of the received lights through the dots in each projected area should be distinct.\footnote{To detect the exact pattern of the coloured dots or the light intensity through each dot, the PD cannot be a simple or common PD but possibly has to be an array of PDs packed or co-located in a particular form and with a tailor-made sophisticated measurement or decoding algorithm for the detection. This would be bulky and expensive.} The printing on the film should ensure that the multiset of the colours of the dots in each projected area is unique. This motivates our construction of universal torus packing for multisets.

\paragraph*{Feasibility and Performance Analysis.}
Let’s denote the total number of colours that we would use for colouring all the dots by $K$. For example, one can consider $K = 4$ by a laser printer and $K = 6$ to $12$ by a inkjet printer.
Note that the typical dots per inch (DPI) for a conventional laser printer is from 600 to 2400, while the typical DPI for a inkjet printer is from 300 to 720.
The distance between two neighboring dots on the film, denoted by $d$ in \cref{fig:Location}, would depend on the DPI value. We denote the number of dots inside the projected area on the film by $N$. We have $N = (v/d)^2$. Since $v = (W \times Z) / h$, 
\begin{equation}
    \label{eq:N}
    N = (v/d)^2 = \left(\frac{(W \times Z) / h}{d}\right)^2 = \left(\frac{W \times Z}{h\times d}\right)^2.
\end{equation}
To justify the above scheme, we need to evaluate how many different multisets can be produced by the projected area, which has $N$ dots, and thus determine whether it is sufficient for identifying each window.

The number of different multisets that can be produced by the projected area is expressible as $C(K,N) = \binom{K+N-1}{N}$.
We consider DPI $= 720$ by printing, \ie, there are $280$ dots/cm, we would have $d = 1/280$  cm $= 0.0357$ mm. Thus, $N = (v/d)^2 = (0.225/0.0357)^2 \approx 40$. Consider $K = 6$, We have $C(6, 40) \approx 1.22 \times 10^6$. Note that today’s conventional light sensor can be equipped with 16-bit ADC \cite{Sensor}, which can distinguish a total of $2.81 \times 10^{14}$ different colours. We can therefore expect that it can  distinguish $C(6, 40)$ different colour multisets since $2.81 \times 10^{14}$  is much larger than $1.22 \times 10^6$.

Consider the angle of view (AoV) of the PD is $60$ degrees \cite{Sensor}. 
Then, the length of the film, denoted by $L$, is to be equal to $2\times(Z\tan(\theta)) = 2\times(10 \tan(60)) \approx 34.6$ mm, by the geometry. The number of possible positions of the projected area on the film is given by $(L/d)^2$, \ie,  $(34.6/0.0357)^2 \approx 9.39 \times 10^5$, which is less than $C(K, N) \approx 1.22 \times 10^6$. Therefore, there is no information theoretic barrier to colouring the dots on the film such that each projected area has different colour multiset for different received light incidence angles. In practice, our algorithm requires a constant factor of elbow room in this calculation, for which we would increase $C(K, N)$ by using more colours or a higher DPI value.

Regarding the possible localization accuracy, we can use the following relationship to estimate by  similar triangles according to \cref{fig:Location}: 
$D / x = Z / h$, so $x =  (h/Z) \times D$. 
Thus we estimate the maximum error as $\Delta x =  (h/Z) \times (\Delta D)$, where $\Delta  D$ is the limit that the PD can detect whether the projected area has changed. $\Delta  D$ is upper bounded by $d$ in our system. Given $h = 2$ meters, $Z = 10$ mm, $d \approx 0.0357$ mm, we have $\Delta x \approx 7.14$ mm. That is, estimation error to the distance $x$ can be less than $1$ cm.

The interesting question of error correction for light based positioning
is arguably yet to be addressed in a manner consistent 
with the physical parameters described above: it is hard to reconcile these parameters
with a camera-based de Bruijn tori model,
given the extremely high precision optics that would be required.
Furthermore, it does not seem appropriate to model errors in the signal
as if it were a bit stream received over a communication channel:
experimental testing may be needed to provide 
a physically realistic model for the errors that need to be corrected.

\subsection{Ambient IoT} 

Positioning has been a fundamental topic
in cellular networks since their inception, {\em e.g.}, from the mid-nineties in meeting requirements for emergency calls~\cite{ericson_blog},
to 5G~\cite{keating2019overview},
and in future 6G networks~\cite{wild2021joint}.

Ambient IoT (\emph{Internet of Things}) refers to
a wireless sensor network
connecting a large number of objects
using low-cost self-powered sensor nodes
(\eg, ambient backscatters, mentioned in \cref{sec:eng}).
We consider a setting composed of
\begin{itemize}
\item
a factory hall of dimension $250 \times 250 m$,
where we deploy $1$ tag every $1 m^2$,
thus resulting in $62,500$ tags,
\item
a moving robot, that acts as an actuator,
\ie, it sends a low power activation signal,
\item
a receptor, that could also be the moving robot itself,
who receives back the ID from the tags nearby via backscattering,
\item
the reader that determines the position of the reference device
based on the tag IDs.
\end{itemize}

A basic solution for the Ambient IoT deployment
would just assume that each tag has its own ID,
which is different from any other deployed tag.
However, if the number of tags is large,
the number of bits used to represent the tag ID is also large,
thus requiring higher transmit power
at the moving robot to enable backscattering at the ambient IoT tags.
The problem that we try to solve is the following:
Can we optimize the ambient IoT tag deployment
such that the power used by the activator
to transmit the activation signal is minimized \cite{Nokia23}?
This motivates the construction presented in this article
of a universal torus packing for multisets.

\begin{figure}
\begin{center}
\resizebox{\textwidth}{!}{\includegraphics{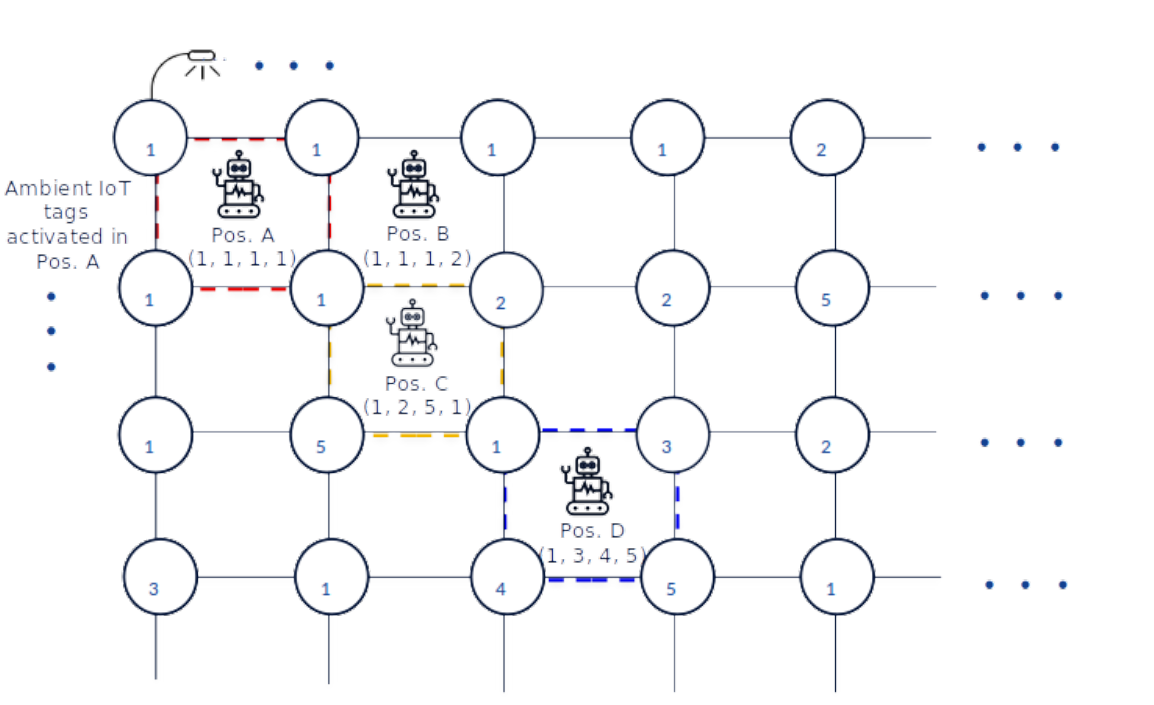}}
\caption{Robot moving in a factory hall, activating tags that send their ID via backscattering \cite{Nokia23}.}
\label{fig:factory:hall}
\end{center}
\end{figure}

For example, consider the factory hall from above.
With $62,500$ tags, we need at least
an average of $15.9$ bits per tag
(as $15.9 < \log_2(62,500)$) to assign different IDs to all the tags.
Using 
the coding described in the proof of \cref{th:construction:grid:colouring},
with parameters $d=2$, $b=2$, $s=1$ for \cref{th:construction:vector sum packing},
we can construct a grid of dimension $2$
and size $256$, with a window of size $8$ and using $5$ colours.
The average number of bits per tag is $\log_2(5) \approx 2.321$,
significantly reducing the power necessary by the activator
to get the ID of the tags via backscattering.

\section{Software Implementation}

We implemented our algorithm constructing the grid colouring
as well as the decoding algorithm in Python.
The code is available at \cite{ourcode}.
An example of grid colouring computed with this code
is presented in \cref{fig:large:grid}.

\begin{figure}[h!]
\begin{center}
\includegraphics[scale=0.3]{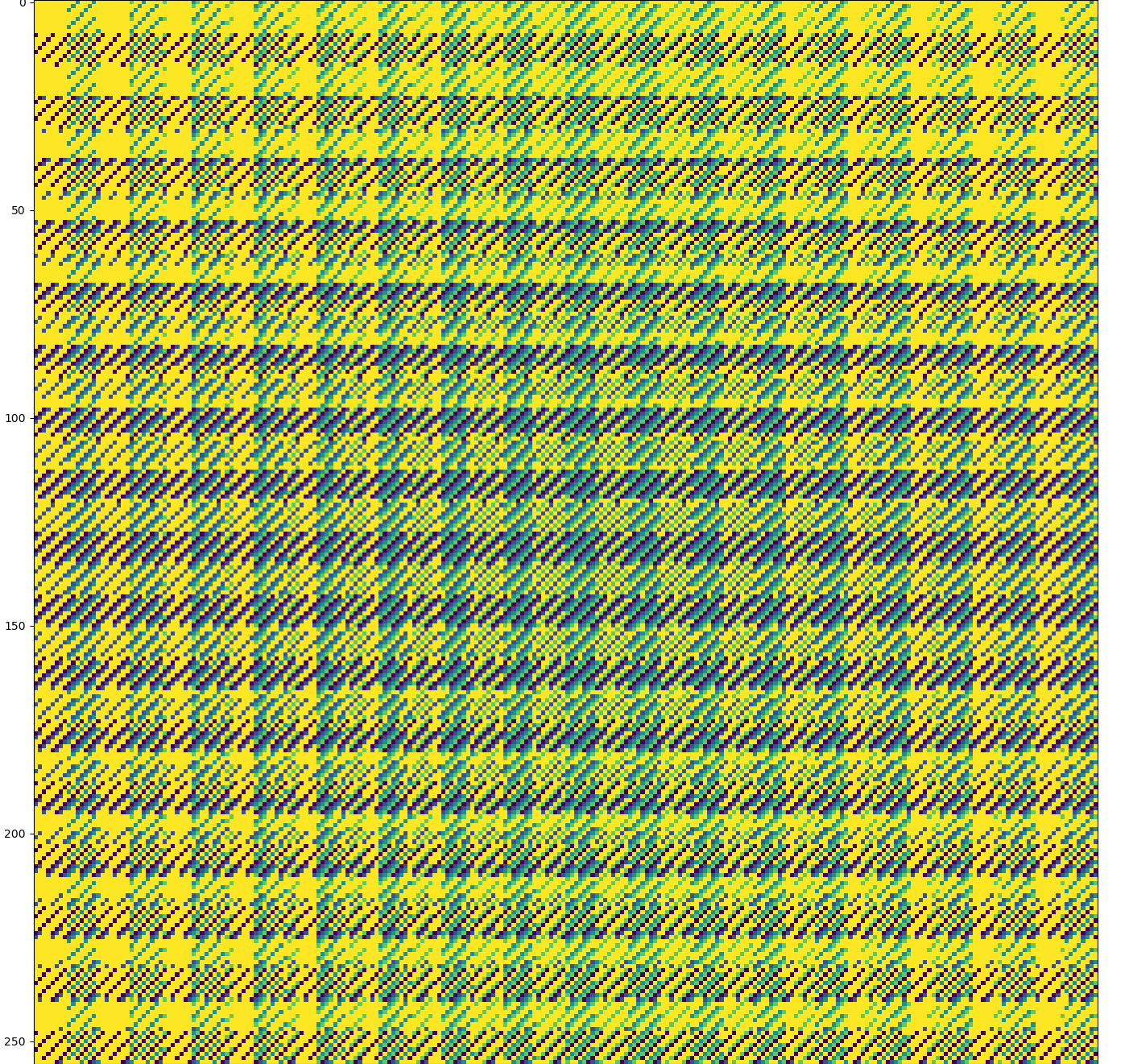}
\caption{Grid colouring
of size $256$, window size $8$ and number of colours $5$.
It corresponds to the parameters $d=2$, $b=2$ and $t=1$.}
\label{fig:large:grid}
\end{center}
\end{figure}

\end{document}